\documentclass{article}

\usepackage[a4paper,margin=2.5cm,top=3.5cm,bottom=3.5cm]{geometry}

\usepackage{amsmath}
\usepackage{amsthm}

\usepackage{hyperref}

\newtheorem{definition}{Definition}
\newtheorem{lemma}[definition]{Lemma}
\newtheorem{theorem}[definition]{Theorem}

\usepackage{tikz}
\usetikzlibrary{calc}
\usetikzlibrary{backgrounds}
\usetikzlibrary{scopes}

\title{A Note on the Faces of the Dual Koch Arrangement}
\date{February 27, 2023}
\author{
  Bernd G\"artner\footnote{Department of Computer Science, ETH Z\"urich, Switzerland. E-mail: \href{mailto:gaertner@inf.ethz.ch}{gaertner@inf.ethz.ch}} \and 
  Manuel Wettstein\footnote{Department of Computer Science, ETH Z\"urich, Switzerland. E-mail: \href{mailto:manuelwe@inf.ethz.ch}{manuelwe@inf.ethz.ch}}
}

\definecolor{colorred}{RGB}{231,76,60}
\definecolor{colorblue}{RGB}{41,128,185}
  
\def\mathlabel#1{\tiny$\mathbf{#1}$}

\tikzset{
  line/.style={rounded corners=6pt},
  baseline/.style={line,thick},
  point/.style={fill,circle,inner sep=1.5pt},
  basepoint/.style={point,inner sep=2pt},
  edgenumber/.style={black,opacity=0.5},
  area/.style={fill,opacity=0.12}
}

\begin{document}

  \maketitle
  
  \begin{abstract}
    We analyze the faces of the dual Koch arrangement, which is the arrangement of $2^s + 1$ lines obtained by projective duality from the Koch chain $K_s$.
    In particular, we show that this line arrangement does not contain any $k$-gons for $k > 5$, and that the number of pentagons is $3 \cdot 2^{s-1} - 3$.
  \end{abstract}
  
  \subparagraph{The Koch chain.}
  The \emph{Koch chain} $K_s$ is a set of $2^s+1$ points in the Euclidean plane, first introduced by Rutschmann and Wettstein \cite{RW22} for the purpose of establishing an improved lower bound on the maximum number of triangulations of planar point sets.
  It can be defined recursively, as follows.
  
  The first iteration of the Koch chain $K_1$, which can be seen on the left hand side of Figure~\ref{fig:chainsmall}, comprises the vertices $p_{-1},p_0,p_{+1}$ of a triangle (where here, and also later, the indices indicate the order of all points along the $x$-axis) with coordinates
  \begin{align*}
    p_{-1} = (-1,0), && p_0 = (0,-1), && p_{+1} = (1,0).
  \end{align*}
  
  To construct $K_s$ for $s > 1$, we again start by placing three vertices $p_{-2^{s-1}},p_0,p_{+2^{s-1}}$ of a triangle with the same coordinates as before, namely
  \begin{align*}
    p_{-2^{s-1}} = (-1,0), && p_0 = (0,-1), && p_{+2^{s-1}} = (1,0).
  \end{align*}
  
  To continue, two copies of $K_{s-1}$ are made arbitrarily flat along the vertical direction, and then translated and rotated in such a way that they come to lie on the edges $p_{-2^{s-1}}p_0$ and $p_0p_{+2^{s-1}}$, respectively, with all points (except for the three vertices we started with) in the interior of the initial triangle.
  The specific construction of $K_2$ can be seen on the right hand side of Figure~\ref{fig:chainsmall}, while the general case is illustrated in Figure~\ref{fig:chaingeneral}.
    
  If the two copies of $K_{s-1}$ have been made sufficiently flat, then any straight-line segment that connects a point from the left copy with a point from the right copy (where we exclude the common point $p_0$) will have no third point in its upper shadow.
  As argued in \cite{RW22}, this yields a unique \emph{order type} (see \cite{GP83} for a definition) with unavoidable edges (that is, edges that are contained in every triangulation of the point set) between any two consecutive points $p_i$, $p_{i+1}$ on the Koch chain.
  
  \begin{figure}
    \centering
    \def\scale{0.4}
    \begin{tikzpicture}[scale=\scale]
      \def\eps{0.25}
      \node[basepoint,label=left:\mathlabel{-1}] (b) at (-5,+5) {};
      \node[basepoint,label=below:\mathlabel{0}] (c) at (0,0) {};
      \node[basepoint,label=right:\mathlabel{+1}] (r) at (+5,+5) {};
      \fill[area] (c.center) -- (b.center) -- (r.center) -- (c.center);
    \end{tikzpicture}
    \hspace{2cm}
    \begin{tikzpicture}[scale=\scale]
      \def\eps{0.25}
      \node[basepoint,label=left:\mathlabel{-2}] (b) at (-5,+5) {};
      \node[basepoint,label=45:\mathlabel{-1}] (bb) at (-2.5+\eps,+2.5+\eps) {};
      \node[basepoint,label=below:\mathlabel{0}] (c) at (0,0) {};
      \node[basepoint,label=135:\mathlabel{+1}] (rr) at (+2.5-\eps,+2.5+\eps) {};
      \node[basepoint,label=right:\mathlabel{+2}] (r) at (+5,+5) {};
      \fill[area] (c.center) -- (b.center) -- (bb.center) -- (c.center);
      \fill[area] (c.center) -- (r.center) -- (rr.center) -- (c.center);
      \fill[area] (c.center) -- (b.center) -- (r.center) -- (c.center);
    \end{tikzpicture}
    \caption{
      On the left, the Koch chain $K_1$ with three points.
      On the right, the Koch chain $K_2$ with five points, where the two flattened and rotated copies of $K_1$ are clearly visible.
    }
    \label{fig:chainsmall}
  \end{figure}
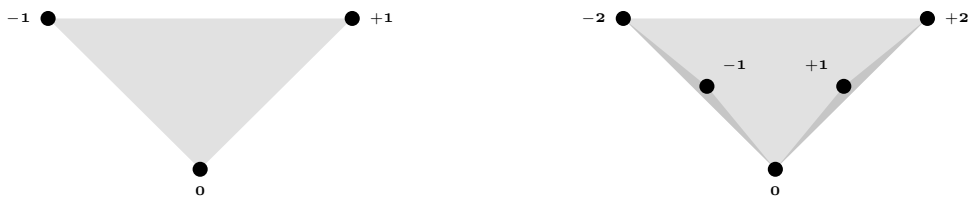
  \begin{figure}
    \centering
    \begin{tikzpicture}[scale=0.8]
      \def\eps{0.2}
      \node[basepoint,colorblue,label=left:\color{colorblue}\mathlabel{-2^{s-1}}] (b) at (-5,+5) {};
      \node[basepoint,colorblue,label=45:\color{colorblue}\mathlabel{-2^{s-2}}] (bb) at (-2.5+\eps,+2.5+\eps) {};
      \node[basepoint,black,label=below:\mathlabel{0}] (c) at (0,0) {};
      \node[basepoint,colorred,label=135:\color{colorred}\mathlabel{+2^{s-2}}] (rr) at (+2.5-\eps,+2.5+\eps) {};
      \node[basepoint,colorred,label=right:\color{colorred}\mathlabel{+2^{s-1}}] (r) at (+5,+5) {};
      \def\midpoints#1#2#3{
        \node[point,#1] at ($0.25*(#2)+0.75*(#3)$) {};
        \node[point,#1] at ($0.5*(#2)+0.5*(#3)$) {};
        \node[point,#1] at ($0.75*(#2)+0.25*(#3)$) {};
      }
      \midpoints{colorblue}{b}{bb}
      \midpoints{colorblue}{bb}{c}
      \midpoints{colorred}{c}{rr}
      \midpoints{colorred}{rr}{r}
      {[on background layer]\fill[area] (c.center) -- (b.center) -- (r.center) -- (c.center);}
      \fill[colorblue,area] (c.center) -- (b.center) -- (bb.center) -- (c.center);
      \fill[colorred,area] (c.center) -- (r.center) -- (rr.center) -- (c.center);
      \node at (-3,2) {$\color{colorblue}\overline{K_{s-1}}$};
      \node at (+3,2) {$\color{colorred}\overline{K_{s-1}}$};
    \end{tikzpicture}
    \caption{
      The Koch chain $K_s$ with $2^s+1$ points, which contains two sufficiently flattened and rotated copies of $K_{s-1}$, denoted by $\overline{K_{s-1}}$.
    }
    \label{fig:chaingeneral}
  \end{figure}
  
  \subparagraph{The dual Koch arrangement.}
  By the standard point-line duality transformation
  \begin{align*}
    && \text{point $p = l^* = (a,b)$} && \longleftrightarrow && \text{line $l = p^* \colon y = ax - b$}, &&
  \end{align*}
  every point set in the plane corresponds to a line arrangement of the same size, and vice versa.
  In particular, this holds true for the Koch chain $K_s$, which thus gives rise to the \emph{dual Koch arrangement}, denoted by $K^*_s$.
  Realizations of $K^*_1$ and $K^*_2$ as $x$-monotone pseudo-line arrangements (this makes it easier to see all the faces more clearly) are given in Figure~\ref{fig:arrangementsmall}, while the general case is illustrated in Figure~\ref{fig:arrangementgeneral}.
  
    \begin{lemma}
      \label{lem:faces}
      The faces of the dual Koch arrangement $K^*_s$ can be categorized as follows:
      \begin{itemize}
        \item The unique infinite face without upper boundary has three edges.
        \item The unique infinite face without lower boundary has two edges.
        \item Among the remaining infinite faces that are unbounded towards the left, there are $s-1$ faces with four edges, while all the others have either two or three edges.
        \item Among the remaining infinite faces that are unbounded towards the right, there are $s-1$ faces with four edges, while all the others have either two or three edges.
        \item Among the remaining finite faces, there are $3 \cdot 2^{s-1} - 2s - 1$ faces with five edges, while all the others have either three or four edges.
      \end{itemize}
    \end{lemma}

  \begin{figure}[b]
    \centering
    \begin{tikzpicture}[scale=0.3,yscale=-1,xscale=0.8]
      \draw[baseline] (-12, 0) node [label=left:\mathlabel{0}]  {} -- (+12, 0) node [label=right:\mathlabel{0}]  {};
      \draw[baseline] (-12,+4) node [label=left:\mathlabel{+1}] {} -- (-2,+4) -- (+10,-2) -- (+12,-2) node [label=right:\mathlabel{+1}] {};
      \draw[baseline] (-12,-2) node [label=left:\mathlabel{-1}] {} -- (-10,-2) -- (+2,+4) -- (+12,+4) node [label=right:\mathlabel{-1}] {};
      
      \node[edgenumber] at (0,1.5) {$3$};
      \node[edgenumber] at (0,-2) {$3$};
      \node[edgenumber] at (0,5) {$2$};
      \node[edgenumber] at (-11,-1) {$2$};
      \node[edgenumber] at (-11,2) {$3$};
      \node[edgenumber] at (11,-1) {$2$};
      \node[edgenumber] at (11,2) {$3$};
    \end{tikzpicture}
    \hspace{1cm}
    \begin{tikzpicture}[scale=0.3,yscale=-1,xscale=0.8]
      \draw[baseline] (-12, 0) node [label=left:\mathlabel{0}]  {} -- (+12, 0) node [label=right:\mathlabel{0}]  {};
      
      \draw[baseline] (-12,+3) node [label=left:\mathlabel{+1}] {} -- (-2,+3) -- (+2,+1) -- (+6,+1) -- (+10,-1) -- (+12,-1) node [label=right:\mathlabel{+1}] {};
      \draw[baseline] (-12,+4) node [label=left:\mathlabel{+2}] {} -- (-2,+4) -- (+2,+2) -- (+10,-2) -- (+12,-2) node [label=right:\mathlabel{+2}] {};
      
      \draw[baseline] (-12,-1) node [label=left:\mathlabel{-1}] {} -- (-10,-1) -- (-6,+1) -- (-2,+1) -- (+2,+3) -- (+12,+3) node [label=right:\mathlabel{-1}] {};
      \draw[baseline] (-12,-2) node [label=left:\mathlabel{-2}] {} -- (-10,-2) -- (-2,+2) -- (+2,+4) -- (+12,+4) node [label=right:\mathlabel{-2}] {};
      
      \node[edgenumber] at (0,0.75) {$5$};
      \node[edgenumber] at (0,-2) {$3$};
      \node[edgenumber] at (0,5) {$2$};
      \node[edgenumber] at (-11,1.5) {$4$};
      \node[edgenumber] at (11,1.5) {$4$};
    \end{tikzpicture}
    \caption{
      On the left, the dual Koch arrangement $K^*_1$ with three lines.
      On the right, the dual Koch arrangement $K^*_2$ with five lines.
      The number of edges of some finite and infinite faces is indicated in gray.
    }
    \label{fig:arrangementsmall}
  \end{figure}
  \begin{figure}
    \centering
    \begin{tikzpicture}[scale=0.8,yscale=-1,xscale=0.8]
      \draw[baseline] (-10, 0) node [label=left:\mathlabel{0}]  {} -- (+10, 0) node [label=right:\mathlabel{0}]  {};
      
      \def\lrline#1#2#3#4{
        \draw[#1,colorred] (-10,+2.0+#2) node [label=left:\mathlabel{#3}] {} -- (-2,+2.0+#2) -- (+2,+0.0+#2) -- (+4,+0.0+#2);
        \draw[#1,colorred] (+8,-0.0-#2) -- (+10,-0.0-#2) node [label=right:\mathlabel{#4}] {};
      }
      
      \lrline{line}{0.1}{}{}
      \lrline{line}{0.2}{}{}
      \lrline{line}{0.3}{}{}
      \lrline{line}{0.4}{}{}
      \lrline{baseline}{0.5}{+2^{s-2}}{+2^{s-2}}
      \lrline{line}{0.6}{}{}
      \lrline{line}{0.7}{}{}
      \lrline{line}{0.8}{}{}
      \lrline{line}{0.9}{}{}
      \lrline{baseline}{1.0}{+2^{s-1}}{+2^{s-1}}
      
      \draw[baseline,colorred] (+4,+1.0) -- (+4.5,+1.0) -- (+7.0,-1.0) -- (+8,-1.0);
      \draw[baseline,colorred] (+4,+0.5) -- (+6.25,+0.5) -- (+7.5,-0.5) -- (+8,-0.5);
      
      \def\rlline#1#2#3#4{
        \draw[#1,colorblue] (+10,+2.0+#2) node [label=right:\mathlabel{#3}] {} -- (+2,+2.0+#2) -- (-2,+0.0+#2) -- (-4,+0.0+#2);
        \draw[#1,colorblue] (-8,-0.0-#2) -- (-10,-0.0-#2) node [label=left:\mathlabel{#4}] {};
      }
      
      \rlline{line}{0.1}{}{}
      \rlline{line}{0.2}{}{}
      \rlline{line}{0.3}{}{}
      \rlline{line}{0.4}{}{}
      \rlline{baseline}{0.5}{-2^{s-2}}{-2^{s-2}}
      \rlline{line}{0.6}{}{}
      \rlline{line}{0.7}{}{}
      \rlline{line}{0.8}{}{}
      \rlline{line}{0.9}{}{}
      \rlline{baseline}{1.0}{-2^{s-1}}{-2^{s-1}}
      
      \draw[baseline,colorblue] (-4,+1.0) -- (-4.5,+1.0) -- (-7.0,-1.0) -- (-8,-1.0);
      \draw[baseline,colorblue] (-4,+0.5) -- (-6.25,+0.5) -- (-7.5,-0.5) -- (-8,-0.5);
      
      \def\subbox#1#2#3#4{
        \draw[ultra thick,#3,opacity=0.6] (#1,-1.5) rectangle (#2,+1.5);
        \node at (0.5*#1+0.5*#2,-2.0) {{\color{#3}#4}};
      }
      
      \subbox{-8}{-4}{colorblue}{$\overline{K^*_{s-1}}$}
      \subbox{+4}{+8}{colorred}{$\overline{K^*_{s-1}}$}
      
      \node[edgenumber] at (0,0.5) {$5$};
      \node[edgenumber] at (0,-1) {$3$};
      \node[edgenumber] at (0,3) {$2$};
      
      \node[edgenumber] at (-9,1) {$4$};
      \node[edgenumber] at (+9,1) {$4$};
      
      \node[edgenumber] at (-6,1) {$3$};
      \node[edgenumber] at (-6,-1) {$2$};
      
      \node[edgenumber] at (+6,1) {$3$};
      \node[edgenumber] at (+6,-1) {$2$};
      
    \end{tikzpicture}
    \caption{
      The dual Koch arrangement $K^*_s$ with $2^s+1$ lines, which contains two flipped copies of the dual Koch arrangement $K^*_{s-1}$.
      The number of edges of some key faces is again indicated in gray.
    }
    \label{fig:arrangementgeneral}
  \end{figure}
  
  \begin{proof}
    The statement of the lemma can be verified for $K^*_1$ (and also $K^*_2$, of course) by a simple but careful inspection of Figure~\ref{fig:arrangementsmall}. 
    
    For larger values of $s$, the statement follows inductively by a careful inspection of all the newly created faces in Figure~\ref{fig:arrangementgeneral}.
    The only tricky part is verifying the formula for the number $N_s$ of finite pentagons in $K^*_s$.
    However, after noting that it must satisfy the recurrence $N_s = 2N_{s-1} + 1 + 2(s-2)$, this boils down to an elementary calculation.
  \end{proof}
  
  \begin{theorem}
    \label{thm:facesprojective}
    For all $s \geq 3$, the line arrangement $K^*_s$ in the projective plane has $3 \cdot 2^{s-1} - 3$ pentagons, while all other faces have either three or four edges.
  \end{theorem}
  \begin{proof}
    We already have $3 \cdot 2^{s-1} - 2s - 1$ pentagons from Lemma~\ref{lem:faces}.
    In addition, when embedding $K^*_s$ in the projective plane, the $2(s-1)$ infinite faces with four edges yield one additional pentagon each.
  \end{proof}
  
  It is important to note that the last step of the reasoning in the proof of Theorem~\ref{thm:facesprojective} fails for the case $s=2$ because, as can be seen on the right hand side of Figure~\ref{fig:arrangementsmall}, the infinite faces of $K^*_2$ with four edges match up with an infinite face with only two edges on the other side, thus combining to a tetragon instead of a pentagon.
  By a sheer stroke of luck, however, $K^*_1$ does not have any infinite faces with four edges at all, and hence the formula from Theorem~\ref{thm:facesprojective} also applies to the case $s=1$.
  
  \subparagraph{An open problem.}
  Given that the Koch chain $K_s$ is the currently best candidate for maximizing the number of planar triangulations, and given that its dual line arrangement $K^*_s$ contains only small faces, one might hope that there is some kind of deeper connection between these two attributes.
  However, we are currently not aware of any such connection and it is unclear to us how to even approach such a question.
  
  \subparagraph{Acknowledgements.}
  The study of the faces of the dual Koch arrangement was initiated by a question posed by Emo Welzl at Gremo's Workshop on Open Problems (GWOP) in Binn, Switzerland, June 2022.
  The question was whether there exist arbitrarily large line arrangements in the projective plane for which the size of the largest face can be bounded by a constant $k$.
  While the dual Koch arrangement indeed answers this question affirmatively and optimally for $k=5$, there already exists at least one other generic construction that achieves the same goal and which is described in Gr\"unbaum's book \cite{Gr72}.
  The authors would like to thank Jean Cardinal, Christoph Grunau and Emo Welzl for helpful discussions during the workshop.

\end{document}